\documentclass[12pt]{elsarticle}

\usepackage{amsmath,amsfonts}
\usepackage{amssymb,amsthm}
\usepackage{makecell}
\usepackage{graphicx}
\usepackage{hyperref}

\hypersetup{
  unicode=true,
  pdftitle={Distributed Matrix Multiplication-Friendly Algebraic Function Fields},
  pdfauthor={Yunlong Zhu, Chang-An Zhao}
}

\newtheorem{theorem}{Theorem}[section]
\newtheorem{lemma}{Lemma}[section]
\newtheorem{definition}{Definition}[section]

\newtheorem{proposition}{Proposition}[section]
\newtheorem{example}{Example}[section]
\newtheorem{remark}{Remark}[section]

\begin{document}
\begin{frontmatter}
\title{Distributed Matrix Multiplication-Friendly Algebraic Function Fields}
\author[sysu]{Yunlong Zhu}
\ead{zhuylong3@mail2.sysu.edu.cn}
\author[sysu,ist]{Chang-An Zhao\corref{cor1}}
\cortext[cor1]{Corresponding author}
\ead{zhaochan3@mail.sysu.edu.cn}

\fntext[fn1]{This work is supported  by  Guangdong Basic and Applied Basic Research Foundation of China~(No. 2025A1515011764).}

\address[sysu]{Department of Mathematics, School of Mathematics, Sun Yat-sen University, Guangzhou 510275, P.R.China.}
\address[ist]{Guangdong Key Laboratory of Information Security Technology, Guangzhou 510006,  P.R.China.}

\date{\today}

\begin{abstract}
In this paper, we introduce distributed matrix multiplication (DMM)-friendly algebraic function fields for polynomial codes and Matdot codes, and present several constructions for such function fields through extensions of the rational function field. The primary challenge in extending polynomial codes and Matdot codes to algebraic function fields lies in constructing optimal decoding schemes. We establish optimal recovery thresholds for both polynomial algebraic geometry (AG) codes and Matdot AG codes for fixed matrix multiplication. Our proposed function fields support DMM with optimal recovery thresholds, while offering rational places that exceed the base finite field size in specific parameter regimes. Although these fields may not achieve optimal computational efficiency, our results provide practical improvements for matrix multiplication implementations. Explicit examples of applicable function fields are provided.
\end{abstract}
\end{frontmatter}

\section{Introduction}
Large-scale matrix multiplication over finite fields plays a central role in numerous algorithms with extensive applications in machine learning, signal processing, and natural language processing. This operation is particularly crucial for code-based cryptosystems \cite{Canteaut}, encoding/decoding codes \cite{Beelen,Khalfaoui}. The computational demands of modern datasets typically exceed the capacity of individual machines. An effective solution involves partitioning the matrices into distinct submatrix blocks and distributing them across multiple worker nodes in a network. After parallel computations, worker nodes return results to the master node, which subsequently decodes the aggregated results to reconstruct the final matrix product. This framework is called Distributed Matrix Multiplication (DMM).

DMM implementations face several primary challenges for the master node. First, communication with numerous worker nodes increases computational overhead, while overall computation time becomes constrained by the slowest node (straggler effect), leading to significantly increased execution time. Recent advancements in straggler mitigation are discussed in \cite{Dutta,Dutta-mat,Fidalgo,Lee,Li-Xing,Yu-poly,Yu-straggler}. Second, security concerns arise regarding potential exposure of sensitive information to worker nodes. For current research on Secure DMM (SDMM), we refer the reader to \cite{Aliasgari,Chang,Karpuk,Machado,Makkonen-secure,Makkonen,Matthews,Oliveira}.

\begin{figure}
	\centering
	\includegraphics[scale=0.5]{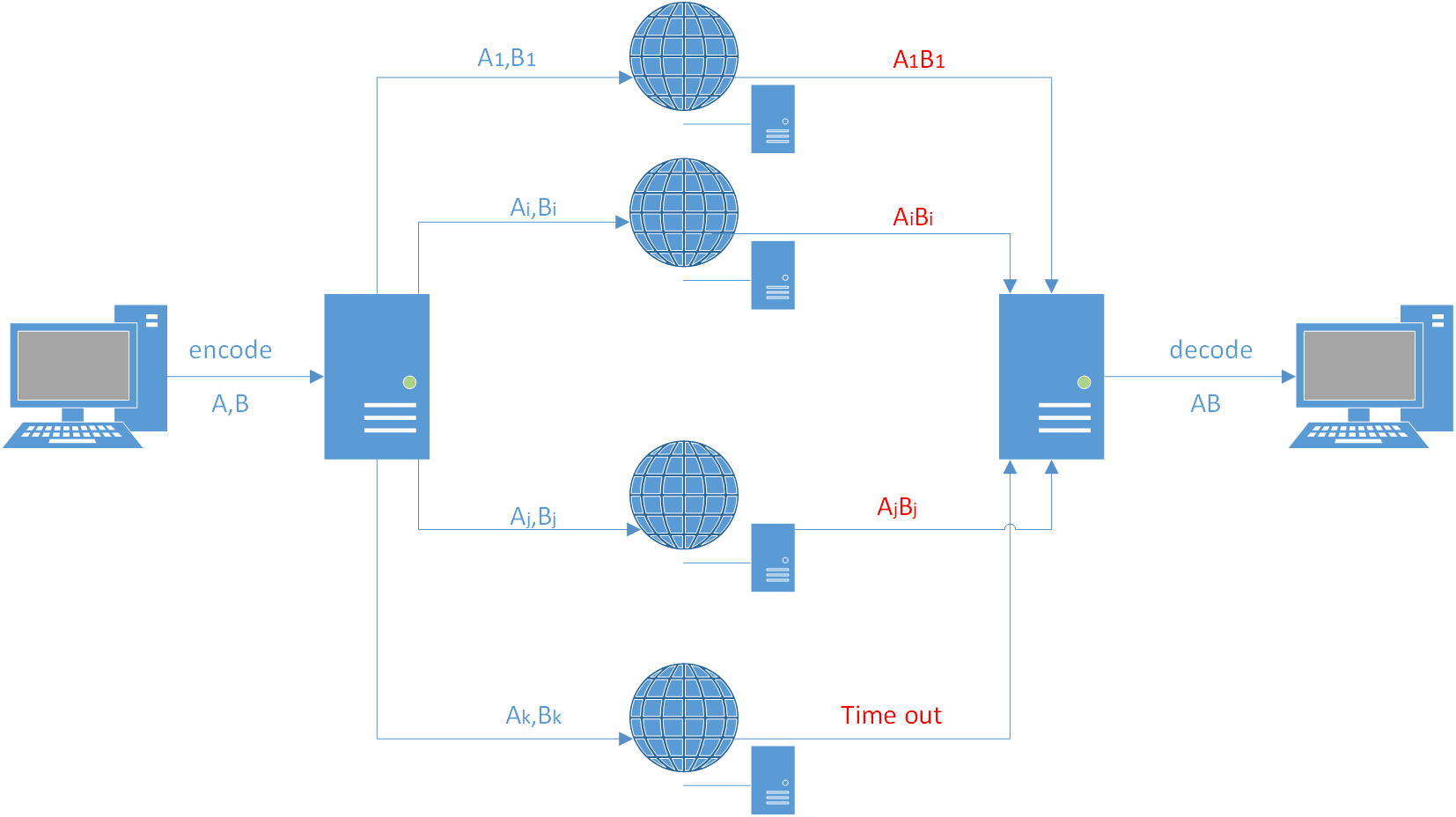}
	\caption{General code-based DMM}
\end{figure}
In this manuscript, we focus on code-based DMM. Given target matrices $A$ and $B$, the master node partitions them into submatrix blocks $A_i$ and $B_j$ respectively. It encodes the computation of $A_iB_j$ using linear error-correcting codes for each pair $(i,j)$. Subsequently, the master node distributes $N$-length codewords to $N$ worker nodes and collects the returned results. Through decoding sufficient partial computations, the master node reconstructs $AB$. Coding theory enables the master node to minimize performance impacts caused by worker node latency.

\subsection*{Related Work}
In \cite{Yu-poly}, a polynomial code DMM was introduced using evaluation codes such as Reed-Solomon codes. Subsequently, new constructions of Matdot and polydot codes were proposed in \cite{Dutta-mat}, also utilizing the rational function field. Compared with polynomial codes, Matdot codes exhibit a better recovery threshold, but worse communication and computational costs per worker node. A common limitation of polynomial codes and Matdot codes is that the number of distinct worker nodes is limited by the field size due to their reliance on the rational function field. To address this constraint, algebraic function fields were employed in \cite{Matthews}, followed by an extension to polynomial algebraic geometry (AG) codes and polydot AG codes in \cite{Fidalgo}. The results in \cite{Fidalgo} were further improved in \cite{Li-Xing}. Notably, algebraic function fields have also been applied in SDMM \cite{Machado,Makkonen}.

Specifically, the authors of \cite{Fidalgo} established a general construction for both polynomial AG codes and Matdot AG codes using a numerical semigroup $S$. Let $c(S)$ denote the conductor of $S$, such that $n\in S$ for all $n\ge c(S)$. Their codes achieve recovery thresholds of $c(S)+mn$ and $2m-1+2c(S)$ in special cases, respectively. They further demonstrated that the optimal recovery threshold for polynomial codes is $g(F)+mn$, where $g(F)$ denotes the genus of the employed function field. However, their analysis was confined to one-point numerical semigroups. In \cite{Li-Xing}, the authors utilized a non-special divisor of degree $g(F)$ to obtain sequences of consecutive pole numbers. Their polynomial AG codes achieve the optimal recovery threshold $g(F)+mn$ when $m\in S$, while their Matdot AG codes attain an explicit recovery threshold of $2g(F)+2m-1$.
\subsection*{Our Contributions}
In this paper, we present some constructions of function fields that enable polynomial AG codes and Matdot AG codes with optimal recovery thresholds using an effective method. We introduce polynomial code-friendly function fields and reprove that the polynomial AG codes maintain a recovery threshold of at least $g(F)+mn$, regardless of divisor selection. These fields offer a novel perspective on the first open problem in \cite{Fidalgo} and a practical approach for DMM. However, we emphasize that these fields may not be optimal for specific multiplication problems. Subsequently, we present two explicit constructions of such fields. For Matdot codes, we prove that the optimal recovery threshold is $2g(F)+2m-1$ when $m \ge g(F)+1$ over a fixed function field $F$. Additionally, we define Matdot code-friendly function fields and provide several constructions. Finally, we discuss limitations of our results.
\subsection*{Organization}
The remainder of this paper is organized as follows. Section II reviews foundational concepts of algebraic function fields and AG codes. In Section III, we establish the polynomial code-friendly function fields and present two constructions. Section IV details our main results on Matdot code-friendly function fields. Section V summarizes a decoding algorithm and computational complexity from prior literature. In Section VI, we conclude this paper with a discussion of limitations and open problems.

\section{Preliminaries and Notations}
This section revisits fundamental definitions and concepts of algebraic function fields and AG codes. Throughout this paper, we denote by $q$ a prime power and by $\mathbb{F}_q$ a finite field. Most content in this section follows \cite{Stich}.
\subsection{Algebraic Function Fields}
Let $q$ be a prime power. An algebraic function field $F/\mathbb{F}_q$ is a finite extension of $\mathbb{F}_q(x)$ with genus $g(F)$, where $x$ is transcendental over $\mathbb{F}_q$. The degree of a place $P$ is defined as
\[
	\deg(P)=[\mathcal{O}_P/P:\mathbb{F}_q].
\]
If $\deg(P)=1$, then $P$ is called rational. Let $N(F)$ denote the number of rational places of $F/\mathbb{F}_q$. The Hasse-Weil bound asserts that
\[
	|q+1-N(F)|\le 2g(\mathcal{X})\sqrt{q}.
\]

A divisor $G$ is defined as a formal sum of places
\[
	G:=\sum\limits_Pn_{P}P
\]
with $n_{P}\in\mathbb{Z}$ and $\deg(G)=\sum\limits_Pn_{P}$. For any function $z\in F$,
\[
	(z):=\sum\limits_Pv_P(z)P=(z)_0-(z)_{\infty}
\]
is called the principal divisor where $v_P$ denotes the discrete valuation at $P$. Given two divisors $G_1,G_2$, we write $G_1\le G_2$ if $n_P(G_1)\le n_P(G_2)$ for all places. The least common multiple (l.c.m) of $G_1$ and $G_2$ is defined as
\[
	\text{\rm l.c.m}(G_1,G_2)=\sum\limits_P\max\{v_P(G_1),v_P(G_2)\}P.
\]
The associated Riemann-Roch space is defined as
\[
	\mathcal{L}(G):=\{z\in F\mid(z)+G\ge0\}\cup\{0\}.
\]
This space forms a vector space over $\mathbb{F}_q$ with dimension denoted by $\ell(G)$. If $\ell(G)=\deg(G)+1-g(F)$, then $G$ is called non-special. Define
\[
	\mathcal{L}(\infty G)=\bigcup_{k\in\mathbb{N}}\mathcal{L}(kG).
\]

An integer $n$ is called a pole number of $P$ if there exists an element $z\in F$ such that $(z)_{\infty}=nP$; otherwise, $n$ is called a gap number. All pole numbers of $P$ form a semigroup denoted by $W(P)$. Using the notation $[0,n]=\{0,1,\ldots,n\}$ and $[n,\infty)=\{n,n+1,\ldots\}$, the Weierstrass Gap Theorem can be stated as follows.
\begin{lemma}\cite[Theorem 1.6.8]{Stich}
Suppose that $P$ is a rational place of a function field $F/\mathbb{F}$. Then there are exactly $g$ gap numbers in $[0,2g(F)-1]$ of $P$. Each $n\in[2g(F),\infty)$ is a pole number of $P$.
\end{lemma}

Let $F'/\mathbb{F}_{q}$ be an algebraic extension of $F/\mathbb{F}_q$ with $[F':F]<\infty$ and $m\ge1$. A place $P'$ of $F'$ is said lying over $P$ if $P\subseteq P'$, denoted by $P'|P$. The ramification index $e(P'|P)$ is defined as the integer satisfying
\[
	v_{P'}(x)=e(P'|P)\cdot v_P(x)
\]for all $x\in F$. If there exists a place $P'$ with $e(P'|P)=[F':F]$, then $P$ is called totally ramified in $F'/F$. The conorm of $P$ is defined as
\[
	\text{Con}_{F'/F}(P):=\sum\limits_{P'|P}e(P'|P)\cdot P'
\]
and
\[
	\text{Con}_{F'/F}\left(\sum\limits_Pn_{P}P\right):=\sum n_P\cdot \text{Con}_{F'/F}(P).
\]
The genus $g(F')$ is determined by the Hurwitz genus formula:
\[
	2g(F')-2=[F':F](2g(F)-2)+\deg(\text{Diff}(F'/F))
\]
where $\text{Diff}(F'/F)$ denotes the different of $F'/F$. We present two extensions relevant to our purpose.
\begin{lemma}\cite[Proposition 6.3.1]{Stich}\label{stich}
Let $F=\mathbb{F}_q(x,y)$ be a function field defined by the equation:
\[
	y^m=a\prod\limits_{i=1}^{\ell}p_i(x)^{n_i}.
\]
with $a\in\mathbb{F}_q$, $\gcd(n,q)=1$, and $\gcd(n,n_i)=1$. Let $d:=\gcd(n,\sum_{i=1}^{\ell}n_i\deg(p_i(x)))$. Then
\begin{itemize}
	\item[(a)] The zero places $P_i$ of $p_i(x)$ are totally ramified in $F/\mathbb{F}_q(x)$. All places $Q_{\infty}$ of $F/\mathbb{F}_q$ satisfying $Q_{\infty}|P_{\infty}$ have ramification index $e(Q_{\infty}|P_{\infty})=\frac{n}{d}$.
	\item[(b)] The genus of $F/\mathbb{F}_q$ is given by
	\[
		g(F)=\frac{(n-1)}{2}\left(-1+\sum_{i=1}^{s}\deg p_i(x)\right)-\frac{d-1}{2}.
	\]
\end{itemize}
\end{lemma}
\begin{lemma}\cite[Proposition 3.1]{Navarro}\label{navarro}
Let $F=\mathbb{F}_q(x,y)$ be a function field defined by the equation:
\[
	\sum\limits_{i=0}^r a_iy^{p^i}=\prod\limits_{i=1}^{\ell}p_i(x)^{n_i}
\]
with $a_r,a_0\neq0$, and $\gcd(q,\sum_{i=1}^{\ell}n_i\deg(p_i(x)))=1$. Let $p_{i_1},\ldots,p_{i_{k}}$ be polynomials with $n_{i_j}<0$. Then
\begin{itemize}
	\item[(a)] The zero places $P_{i_j}$ of $p_{i_j}(x)$ and the pole place $P_{\infty}$ of $x$ are totally ramified in $F/\mathbb{F}_q(x)$.
	\item[(b)] Let $Q_{i_j}$ and $Q_{\infty}$ be places of $F/\mathbb{F}_q$ with $Q_{i_j}|P_{i_j}$ and $Q_{\infty}|P_{\infty}$. Then
	\[
		(p_{i_j})=q^rQ_{i_j}-q^r\deg(p_{i_j})Q_{\infty}
	\]
	and
	\[
		(y)_{\infty}=\sum\limits_{j=1}^k-n_{i_j}Q_{i_j}+Q_{\infty}.
	\]
\end{itemize}
\end{lemma}
\subsection{Algebraic Geometry Codes}
For a given divisor $G$ of $F/\mathbb{F}_q$, let $P_1,\ldots,P_n$ be $n$ pairwise distinct rational places of $F/\mathbb{F}_q$ with $P_i\notin \text{\rm supp}(G)$ for all $i$. Let $D=P_1+\cdots+P_n$, and consider the evaluation map:
\begin{align*}
	\text{ev}_D:  \mathcal{L}(G)&\to \mathbb{F}_q^n,\\
	f&\mapsto (f(P_1),\ldots,f(P_n)).
\end{align*}
The AG code denoted by $C_L(D,G)$ represents the image of $\text{ev}_D$. The parameters of $C_L(D,G)$ are given by:
\[
	k=\ell(G)-\ell(G-D), d\ge n-\deg(G),
\]
where $n-\deg(G)$ is defined as the design distance $d^*$ of $C_L(D,G)$. When $\deg(G)<n$ it follows directly that $\text{ev}_D$ constitutes an embedding and $k=\ell(G)$.

Let $G_{C_L}$ denote the generator matrix of $C_L$. Then $G_{C_L}\in\mathbb{F}_q^{\ell(G)\times n}$ has row rank $k$. Moreover, if $\{f_1,\ldots,f_{k}\}$ forms a basis for $\mathcal{L}(G)$, then
\begin{equation*}
	G_{C_L}=\begin{pmatrix}
		f_1(P_1)&f_1(P_2) &\ldots&f_1(P_n)\\
		f_2(P_1)&f_2(P_2)&\ldots&f_2(P_n)\\
		\vdots&\vdots&\ddots&\vdots\\
		f_{k}(P_1)&f_{k}(P_2)&\ldots&f_{k}(P_n)
	\end{pmatrix}
\end{equation*}
\section{Polynomial Code-Friendly Algebraic Function Fields}
We first review classical polynomial codes proposed in \cite{Yu-poly}, then introduce polynomial code-friendly function fields and present two constructive families.  
\subsection{Classical Polynomial AG Codes}
Consider matrices $A\in\mathbb{F}_q^{r\times s}$ and $B\in\mathbb{F}_q^{s\times t}$ of size $r\times s$ and $s\times t$. Suppose that $m$ and $n$ are divisors of $r$ and $t$ respectively. Partition these matrices into submatrices:
\begin{equation*}
	A=\begin{pmatrix}
		A_1\\
		A_2\\
		\vdots\\
		A_m
	\end{pmatrix},B=\begin{pmatrix}
	B_1,B_2,\ldots,B_n
\end{pmatrix}
\end{equation*}
where $A_i\in\mathbb{F}_q^{\frac{r}{m}\times s}$ and $B_j\in\mathbb{F}_q^{s\times \frac{t}{n}}$. Consequently, the product $AB$ decomposes as:
\begin{equation*}
	AB=\begin{pmatrix}
		A_1B_1&A_1B_2&\ldots&A_1B_n\\
		A_2B_1&A_2B_2&\ldots&A_2B_n\\
		\vdots&\vdots&\ddots&\vdots\\
		A_mB_1&A_mB_2&\ldots&A_mB_n
	\end{pmatrix}.
\end{equation*}
Let $F/\mathbb{F}_q$ be an algebraic function field and $G$ a divisor of $F/\mathbb{F}_q$. Suppose that $\{f_1,\ldots,f_m\}\subset\mathcal{L}(\infty G)$ and $\{g_1,\ldots,g_n\}\subset\mathcal{L}(\infty G)$ are linearly independent function sets. The master node selects two matrix-coefficient polynomials:
\[
	f:=\sum\limits_{i=1}^mA_if_i,\ g:=\sum\limits_{j=1}^nB_jg_j,
\]
satisfying the following condition:
\begin{itemize}
	\item $v_Q(f_ig_j)\neq v_Q(f_{i'}g_{j'})\ \text{\rm if}\ (i,j)\neq(i',j')$.
\end{itemize}
The condition ensures the linear independence of the $mn$ functions $f_ig_j$. Consequently, each submatrix $A_iB_j$ can be recovered as the coefficient of
\[
	h:=fg=\sum\limits_{i=1}^m\sum\limits_{j=1}^nA_iB_jf_ig_j.
\]
corresponding to the monomial $f_ig_j$. Clearly $h\in\mathcal{L}(\infty G)$. 

To compute $AB$ with parallelization and straggler resistance, the master node proceeds as follows: First, select $N$ distinct places $\{P_1,\ldots,P_N\}$ in $F/\mathbb{F}_q$ such that $P_i\notin{\rm supp}(G)$. Then distribute $f(P_i)$ and $g(P_i)$ to worker nodes. Each worker node computes $f(P_i)g(P_i)=h(P_i)$ and returns the result. After retrieving enough products, the master node can reconstruct $AB$ by recovering $h$ through interpolation.

Assume that $h\in\mathcal{L}(kG)^{\frac{r}{m}\times\frac{t}{n}}$ for some $k$. The matrix vectors $(h(P_1),\ldots,h(P_N))$ exactly contain codewords in $C_L(D,kG)$ at each $[i,j]$-th entry position. The code $C_L(D,kG)$ has a design distance $d^*=N-\deg(kG)$, enabling error correction for up to $\lfloor\frac{d^*-1}{2}\rfloor$ straggler worker nodes.

The recovery threshold $R$ is defined as the minimum number of worker node responses needed to recover $h$, which satisfies
\[
	R=\deg((h)_{\infty})+1
\]
through the decoding algorithm in \cite{Fidalgo} (see Section V). This implies requiring at least $R$ evaluation places.
\begin{remark}
In \cite{Yu-poly}, the rational function field $\mathbb{F}_q(x)/\mathbb{F}_q$ and $G=P_{\infty}$, the place at infinity yield the construction
\[
	f_i=x^{i-1},\ g_j=x^{m(j-1)}
\]
and $R=\deg(h)+1=mn$. This construction achieves the optimal recovery threshold provided in \cite[Theorem 1]{Yu-poly}.
\end{remark}
\subsection{DMM-Friendly Algebraic Function Fields}
We extend Proposition 1 in \cite{Fidalgo} to multi-point codes through the following result:
\begin{proposition}\label{op}
Let $R^*$ denote the optimal recovery threshold for polynomial AG codes over $F/\mathbb{F}_q$, defined as the minimum achievable recovery threshold among all computation strategies over $F/\mathbb{F}_q$:
\[
	R^*=\min\{\deg((h)_{\infty})+1|h=fg,f,g\in F\}.
\]
Then $R^*\ge g(F)+mn$.
\end{proposition}
\begin{proof}
To recover $AB$, the master node requires linearly independent sets $\{f_i\}$ and $\{g_j\}$ such that $\{f_ig_j\}$ are linearly independent. Suppose that $(f_i)_{\infty}=D_i$ and $(g_j)_{\infty}=E_j$. Then
\[
	(h)_{\infty}= D := \text{\rm l.c.m}(D_i,E_j|1\le i\le m,1\le j\le n)
\]
where $h=fg$. Since the basis of $\mathcal{L}(D)$ contains at least $mn$ elements, the Riemann-Roch theorem implies that $\deg(D)\ge mn+g(F)-1$. Consequently, the recovery threshold is at least $mn+g(F)$ for all possible $f$ and $g$.
\end{proof}
Subsequently, we introduce the following definition.
\begin{definition}
An $(m,n)$-friendly algebraic function field is a function field $F/\mathbb{F}_q$ with a place $Q$ of $F/\mathbb{F}_q$ such that
\begin{itemize}
	\item the genus satisfies $(m-1)|g(F)$,
	\item there exists $z_1\in F$ with $v_Q(z_1)=1$ and $\deg((z_1)_{\infty})=\frac{g(F)}{m-1}+1$,
	\item there exists $z_2$ such that $(z_2)_{\infty}=mQ$.
\end{itemize}
\end{definition}
Then we have
\begin{theorem}
Suppose that $F/\mathbb{F}_q$ is an $(m,n)$-friendly algebraic function field. Let $f_i=z_1^{i-1}$ for $i=1,\ldots,m$ and $g_j=z_2^{j-1}$ for $j=1,\ldots,n$. Then the functions
\[
	f:=\sum\limits_{i=1}^mA_if_i,\ g:=\sum\limits_{j=1}^nB_jg_j,
\]
satisfy
\[
	v_Q(f_ig_j)\neq v_Q(f_{i'}g_{j'})\ \text{\rm if}\ (i,j)\neq(i',j').
\]
Consequently, they form a polynomial AG code with recovery threshold $R=g(F)+mn$.
\end{theorem}
\begin{proof}
First, we have $v_Q(f_ig_j)=i-1+m(j-1)$. Suppose that there exist distinct pairs $(i,j)\neq(i',j')$ such that $i-1+m(j-1)=i'-1+m(j'-1)$. This implies $i-i'=m(j'-j)$, leading to a contradiction since $1\le i,i'\le m$. Thus the $f_ig_j$ are $mn$ linearly independent functions that form a polynomial AG code. The recovery threshold $R$ is computed as:
\begin{align*}
	R&=\deg((h)_{\infty})+1\\
	&=(m-1)\deg((z_1)_{\infty})+(n-1)\deg((z_2)_{\infty})\\
	&=(m-1)\left(\frac{g(F)}{m-1}+1\right)+(n-1)m\\
	&=g(F)+mn.
\end{align*}
\end{proof}
\begin{remark}
For any $(m,n)$-friendly algebraic function field $F/\mathbb{F}_q$, this construction naturally provides an effective method to obtain polynomial AG codes with optimal recovery threshold. This result agrees with \cite{Yu-poly} where $g(\mathbb{F}_q(x))=0$.
\end{remark}
\subsection{Constructions of DMM-Friendly Algebraic Function Fields}
Let $\ell\ge2$ be a positive integer. Suppose that $a_1,\ldots,a_{\ell},b_1,\ldots,b_{\ell-1}$ are distinct elements of $\mathbb{F}_q$. We present the following result:
\begin{theorem}\label{poly}
Let $F=\mathbb{F}_q(x,y)$ be a function field defined by the equation:
\[
	y^m=\frac{\prod_{i=1}^{\ell}(x-a_i)}{\prod_{i=1}^{\ell-1}(x-b_i)}.
\]
Then $F/\mathbb{F}_q$ is an $(m,n)$-friendly algebraic function field.
\end{theorem}
\begin{proof}
Let $P_{\alpha}$ be the place associated with $(x-\alpha)$, and $P_{\infty}$ the pole of $x$ in $\mathbb{F}_q(x)$. By Lemma \ref{stich}, each place $P_{\alpha}$ for $\alpha\in\{a_1,\ldots,a_{\ell},b_1,\ldots,b_{\ell}\}$ is totally ramified in $F/\mathbb{F}_q(x)$, with corresponding place $Q_{\alpha}|P_{\alpha}$ in $F/\mathbb{F}_q$. The place $Q_{\infty}$ lying over $P_{\infty}$ has ramification index $e(Q_{\infty}|P_{\infty})=m$, hence it is also totally ramified. It follows that
\begin{align*}
	(y^m)&=\text{Con}_{F/\mathbb{F}_q(x)}\left(\sum\limits_{i=1}^{\ell}P_{a_i}-\sum\limits_{i=1}^{\ell-1}P_{b_i}-P_{\infty}\right)\\
	&=m\left(\sum\limits_{i=1}^{\ell}Q_{a_i}-\sum\limits_{i=1}^{\ell-1}Q_{b_i}-Q_{\infty}\right),
\end{align*}
which implies that
\[
	(y)=\sum\limits_{i=1}^{\ell}Q_{a_i}-\sum\limits_{i=1}^{\ell-1}Q_{b_i}-Q_{\infty}
\]
with $v_{Q_{\infty}}(y)=1$. Moreover, we have
\[
	(x)_{\infty}=\text{Con}_{F/\mathbb{F}_q(x)}(P_{\infty})=mQ_{\infty}.
\]
The genus of $F$ is 
\[
	g(F)=(\ell-1)(m-1).
\]
Consequently, we obtain
\[
	\deg((y)_{\infty})=\ell=\frac{g(F)}{m-1}+1.
\]
Therefore $F/\mathbb{F}_q$ satisfies the conditions for an $(m,n)$-friendly algebraic function field.
\end{proof}
A special case occurs when $m$ is a power of $p$, specifically $m=p^{u-1}$. We present the following result.
\begin{theorem}\label{linear}
Let $F=\mathbb{F}_q(x,y)$ be a function field defined by the equation:
\[
	{\rm Tr}_{p^{u}/p}(y)=\frac{\prod_{i=1}^{\ell}(x-a_i)}{\prod_{i=1}^{\ell-1}(x-b_i)},
\]
where ${\rm Tr}_{p^{u}/p}(y)$ denotes the trace map of $\mathbb{F}_{p^{u}}$ over $\mathbb{F}_{p}$ with degree $p^{u-1}$. Then $F/\mathbb{F}_q$ is an $(m,n)$-friendly algebraic function field.
\end{theorem}
\begin{proof}
By Lemma \ref{navarro}, the places $P_{b_i}$ and $P_{\infty}$ are totally ramified and they are the only ramified places in $F/\mathbb{F}_q(x)$. It follows that
\[
	(y^{p^{u-1}})_{\infty}=\text{Con}_{F/\mathbb{F}_q(x)}(\sum\limits_{i=1}^{\ell-1}P_{b_i}+P_{\infty}),
\]
which implies
\[
	(y)_{\infty}=\sum\limits_{i=1}^{\ell-1}Q_{b_i}+Q_{\infty}.
\]
Furthermore, we have
\[
	(x)_{\infty}=\text{Con}_{F/\mathbb{F}_q(x)}(P_{\infty})=mQ_{\infty}.
\]
The different exponents of all $Q_{b_i}$ and $Q_{\infty}$ satisfy
\[
	d(Q_{b_i}|P_{b_i})=d(Q_{\infty}|P_{\infty})=2(q-1)
\]
since $m_{P_{b_i}}=-v_{P_{b_i}}(\frac{1}{x-b_i})$ in Proposition 3.7.10 \cite{Stich}. Thus, from the Hurwitz genus formula, we have
\begin{align*}
	2g(F)-2&=-2q+\deg\text{Diff}(F/\mathbb{F}_q(x))\\
	&=-2q+\sum\limits_{i=1}^{\ell-1}d(Q_{b_i}|P_{b_i})+d(Q_{\infty}|P_{\infty})\\
	&=2\ell(q-1)-2q.
\end{align*}
This yields
\[
	g(F)=(\ell-1)(p^{u-1}-1).
\]
Consequently, we obtain
\[
	\deg((y)_{\infty})=\ell=\frac{g(F)}{m-1}+1.
\]
Therefore $F/\mathbb{F}_q$ satisfies the conditions for an $(m,n)$-friendly algebraic function field.
\end{proof}
\subsection{Comparisons and Examples}
We summarize existing constructions of polynomial AG codes from \cite{Fidalgo} and \cite{Li-Xing} in Table I. 
\begin{table}[htbp]\label{table1}
	\renewcommand{\arraystretch}{1.5}
	\centering
	\caption{Recovery thresholds of AG-based polynomial codes in \cite{Fidalgo} and \cite{Li-Xing}}
	\begin{tabular}{|c|c|c|}
		\hline
		Ref. & $m\notin W(Q)$ & $m\in W(Q)$\\
		\hline
		Construction 2 in \cite{Fidalgo} & 2$c(Q)+mn$ & 2$c(Q)+mn$ \\
		\hline
		Construction 3 in \cite{Fidalgo} & $c(Q)+m'n$ & $c(Q)+mn$ \\
		\hline
		Construction 4 in \cite{Fidalgo} & $c(Q)+m_n+n$ & $c(Q)+mn$\\
		\hline
		Construction 5 in \cite{Li-Xing} & $g(F)+m'n$ & $g(F)+mn$\\
		\hline
		Construction 6 in \cite{Li-Xing} & $g(F)+m_n+n$ & $g(F)+mn$\\
		\hline
	\end{tabular}
\end{table}

Applying Proposition \ref{op}, we observe that the constructions of polynomial AG codes in \cite{Li-Xing} are optimal for general function fields when $m\in W(Q)$. However, identifying divisors $Q$ with $m \in W(Q)$ is non-trivial. Our $(m,n)$-friendly algebraic function fields offer a practical computational alternative. We provide explicit curve examples computed via SageMath \cite{sage}, with deeper analysis reserved in Section VI.
\begin{example}\label{exp1}
Let $q=25$ and $F=\mathbb{F}_q(x,y)$ be defined by
\[
	y^8=\frac{x^2-3}{x-1}.
\]
Then the places $P_{1}$ and $P_{\infty}$ are totally ramified in $F/\mathbb{F}_q(x)$. Moreover, we have
\[
	(y)_{\infty}=Q_{1}+Q_{\infty},(x)_{\infty}=8Q_{\infty},
\]
with $g(F)=7$, and the number of rational places of $F/\mathbb{F}_q$ is 52.
\end{example}
\begin{example}
Let $q=27$ and $F=\mathbb{F}_q(x,y)$ be defined by
\[
	y^{9}+y^3+y=\frac{(x+\theta)(x-\theta)}{x-1}
\]
where $\theta$ is a primitive element. Then the places $P_{1}$ and $P_{\infty}$ are totally ramified in $F/\mathbb{F}_q(x)$. Furthermore, we have
\[
	(y)_{\infty}=Q_{1}+Q_{\infty},(x)_{\infty}=9Q_{\infty},
\]
with $g(F)=8$, and the number of rational places of $F/\mathbb{F}_q$ is 56.
\end{example}
\section{Matdot Code-Friendly Algebraic Function Fields}
\subsection{Classical Matdot AG Codes}
The Matdot codes introduced in \cite{Dutta-mat} split matrices $A$ and $B$ using an alternative method. Consider two matrices $A\in\mathbb{F}_q^{r\times s}$ and $B\in\mathbb{F}_q^{s\times t}$. Suppose that $m$ is a divisor of $s$. We partition
\begin{equation*}
	A=\begin{pmatrix}
		A_1,A_2,\ldots,A_m
	\end{pmatrix},B=\begin{pmatrix}
		B_1\\
		B_2\\
		\vdots\\
		B_m
	\end{pmatrix}
\end{equation*}
with $A_i\in\mathbb{F}_q^{r\times \frac{s}{m}}$ and $B_j\in\mathbb{F}_q^{\frac{s}{m}\times t}$. Consequently, the product matrix $AB$ is of the form
\begin{equation*}
	AB=\sum\limits_{i=1}^mA_iB_i.
\end{equation*}
Following polynomial codes, the master node chooses two matrix-coefficient polynomials:
\[
	f:=\sum\limits_{i=1}^mA_if_i,\ g:=\sum\limits_{j=1}^mB_jg_j
\]
satisfying:
\begin{itemize}
	\item There exists an integer $d$ such that exactly $m$ pairs $f_ig_i$ satisfy $v_Q(f_ig_i)=d$.
\end{itemize}
Define the product
\[
	h:=fg=\sum\limits_{i=1}^m\sum\limits_{j=1}^mA_iB_jf_ig_j.
\]
The specified condition ensures that the degree-$d$ monomial term of $h$ contains $AB$ as its coefficient.

The algorithm proceeds as follows: The master node chooses $N$ distinct places $\{P_1,\ldots,P_N\}$ of $F/\mathbb{F}_q$ such that $P_i\notin{\rm supp}(G)$, and distributes $f(P_i)$ and $g(P_i)$ to worker nodes. Each worker node computes $f(P_i)g(P_i)=h(P_i)$ and returns the result. The master node then reconstructs $AB$ by interpolating $h$.

From the decoding algorithm in \cite{Fidalgo} (see Section V), the recovery threshold $R$ satisfies
\[
	R=\deg((h)_{\infty})+1,
\]
indicating the minimum number of required responses.
\begin{remark}
In \cite{Dutta-mat}, the authors considered the rational function field $\mathbb{F}_q(x)/\mathbb{F}_q$ and $G=P_{\infty}$. Their construction is given by
\[
	f_i=x^{i-1},\ g_j=x^{m-j}
\]
and $R=\deg(h)+1=2m-1$, which is optimal. However, the optimality of Matdot AG codes depends on the employed function field for a given $m$. If $m\ge g(F)+1$, we will demonstrate later that $R$ is at least $2g(F)+2m-1$. The work in \cite{Fidalgo} establishes optimal thresholds for sparse semigroups, subsequently improved for specific curves in \cite{Li-Xing}.
\end{remark}
\subsection{DMM-Friendly Algebraic Function Fields}
We first establish the following proposition.
\begin{proposition}
Let $R^*$ denote the optimal recovery threshold of Matdot AG codes over $F/\mathbb{F}_q$, defined as
\[
	R^*=\min\{\deg((h)_{\infty})+1|h=fg,f,g\in F\}.
\]
If $m\ge g(F)+1$, then $R^*\ge 2g(F)+2m-1$.
\end{proposition}
\begin{proof}
The condition for Matdot AG codes requires the master node to select functions $f_1,\ldots,f_m$ such that $v_Q(f_i)$ form a sequence of consecutive integers at a given place $Q$. As there are exactly $g(F)$ gap numbers in the interval $[0,2g(F)-1]$, the master node must identify at least $m-g(F)$ functions $f'_i$ satisfying $-v_Q(f'_i)\ge 2g(F)$. Thus
\[
	\max\{\deg((f_i)_{\infty})|i=1,\ldots,m\}\ge m-g+2g-1=m+g-1,
\]
which implies that
\[
	\deg((h)_{\infty})\ge 2g(F)+2m-2.
\]
Consequently, the optimal recovery threshold satisfies $R^*=\deg((h)_{\infty})+1\ge2g(F)+2m-1$.
\end{proof}
Subsequently, we introduce the following definition.
\begin{definition}
An $m$-friendly algebraic function field is a function field $F/\mathbb{F}_q$ with a place $Q$ of $F/\mathbb{F}_q$ such that
\begin{itemize}
	\item genus $g(F)=m-1$,
	\item there exists $z_1\in F$ with $v_Q(z_1)=1$ and $\deg((z_1)_{\infty})=2$.
\end{itemize}
\end{definition}
Then we have
\begin{theorem}
Suppose that $F/\mathbb{F}_q$ is an $m$-friendly algebraic function field. Let $f_i=z_1^{i-1}$ for $i=1,\ldots,m$ and $g_j=z_1^{m-j}$ for $j=1,\ldots,m$. Then the functions
\[
	f:=\sum\limits_{i=1}^nA_if_i,\ g:=\sum\limits_{j=1}^nB_jg_j,
\]
satisfy
\[
	v_Q(f_ig_j)=m-1\ \text{\rm if and only if}\ i=j.
\]
Consequently, they form a Matdot AG code with recovery threshold $R=g(F)+mn$.
\end{theorem}
\begin{proof}
Observe that $v_Q(f_ig_j)=i-1+m-j$. Therefore there are exactly $m$ pairs $(f_i,g_i)$ such that $v_Q(f_ig_i)=m-1$. This implies that the functions $f$ and $g$ form a Matdot AG code. The recovery threshold $R$ is computed as follows:
\begin{align*}
	R&=\deg((h)_{\infty})+1\\
	&=2(m-1)\deg((z_1)_{\infty})+1\\
	&=4(m-1)+1\\
	&=2g(F)+2m-1.
\end{align*}
\end{proof}
\subsection{Constructions of DMM-Friendly Algebraic Function Fields}
Suppose that $a_1,a_2,a_3$ are distinct elements of $\mathbb{F}_q$. We present the following result.
\begin{theorem}
Let $F_1=\mathbb{F}_q(x,y_1)$ and $F_2=\mathbb{F}_q(x,y_2)$ be function fields defined by the equations:
\[
	y_1^m=\frac{(x-a_1)(x-a_2)}{x-a_3},
\]
and
\[
	y_2^m=\frac{x-a_3}{(x-a_1)(x-a_2)},
\]
respectively. Then both $F_1/\mathbb{F}_q$ and $F_2/\mathbb{F}_q$ are $m$-friendly algebraic function fields.
\end{theorem}
\begin{proof}
Following the proof of Theorem \ref{poly}, the places $P_{a_3}$ and $P_{\infty}$ are totally ramified in $F_1/\mathbb{F}_q(x)$ while the places $P_{a_1}$ and $P_{a_2}$ are totally ramified in $F_2/\mathbb{F}_q(x)$. It follows that
\[
	(y_1)_{\infty}=Q_{a_3}+Q_{\infty}
\]
and
\[
	(y_2)_{\infty}=Q_{a_1}+Q_{a_2}.
\]
Both function fields have genus $g(F_1)=g(F_2)=m-1$. Therefore $F_1/\mathbb{F}_q$ and $F_2/\mathbb{F}_q$ satisfy the conditions for being $m$-friendly.
\end{proof}
\begin{remark}
While these function field families correspond to the general Kummer extension, their presented forms avoid finite field inversions, enhancing computational efficiency.
\end{remark}
For the case $m=p^{u-1}$, we present the following result.
\begin{theorem}
Let $F_1=\mathbb{F}_q(x,y_1)$ and $F_2=\mathbb{F}_q(x,y_2)$ be function fields defined by the equations:
\[
	{\rm Tr}_{p^{u}/p}(y_1)=\frac{(x-a_1)(x-a_2)}{x-a_3},
\]
and
\[
{\rm Tr}_{p^{u}/p}(y_2)=\frac{x-a_3}{(x-a_1)(x-a_2)},
\]
respectively. Then both $F_1/\mathbb{F}_q$ and $F_2/\mathbb{F}_q$ are $m$-friendly algebraic function fields.
\end{theorem}
\begin{proof}
Following the proof of Theorem \ref{linear}, the places $P_{a_3}$ and $P_{\infty}$ are totally ramified in $F_1/\mathbb{F}_q(x)$ while the places $P_{a_1}$ and $P_{a_2}$ are totally ramified in $F_2/\mathbb{F}_q(x)$. It follows that
\[
	(y_1)_{\infty}=Q_{a_3}+Q_{\infty}.
\]
and
\[
	(y_2)_{\infty}=Q_{a_1}+Q_{a_2}.
\]
Furthermore, the Hurwitz genus formula yields
\[
	g(F_1)=g(F_2)=p^{u-1}-1.
\]
Therefore $F_1/\mathbb{F}_q$ and $F_2/\mathbb{F}_q$ satisfy the conditions for being $m$-friendly.
\end{proof}
\subsection{Comparisons and Examples}	
The recovery threshold bound provided in \cite{Fidalgo} is $2(d-\delta)+1$, where $d$ and $\delta$ are entirely determined by the Weierstrass semigroup at a single place. This bound was subsequently improved to $2m+2g(F)-1$ in \cite{Li-Xing} for general function fields. We have shown that $2m+2g(F)-1$ is optimal when $m>g(F)$ and that it typically holds since $g(F)$ is often sufficiently small. With our constructions of function fields, we can also obtain Matdot AG codes with optimal recovery threshold. We provide explicit curve examples computed via SageMath \cite{sage}, with deeper analysis reserved in Section VI.
\begin{example}
Let $q=25$ and $F=\mathbb{F}_q(x,y)$ be defined by
\[
	y^3=\frac{x}{x^2-3}.
\]
Then the places $P_{1},P_{\infty}$ are totally ramified in $F/\mathbb{F}_q(x)$. Furthermore, we have
\[
	(y)_{\infty}=Q_{3\theta+1}+Q_{2\theta+4},
\]
where $\theta$ is a primitive element of $\mathbb{F}_q$. The genus $g(F)=2$ and the number of rational places of $F/\mathbb{F}_q(x,y)$ is 46.
\end{example}
\section{Decoding and Complexity}
Suppose that the master code has collected sufficiently many evaluations $\{h(P_1),\ldots,h(P_R)\}$. Each value is exactly a matrix in $\mathcal{L}(kG)^{a\times b}$, where $(a,b)=(\frac{r}{m},\frac{t}{n})$ for polynomial AG codes and $(a,b)=(r,t)$ for Matdot AG codes. Let $h_{i,j}$ denote the $(i,j)$-th entry of $h$ for $(i,j)\in[1,a]\times[1,b]$. Then $(h_{i,j}(P_1),\ldots,h_{i,j}(P_N))$ constitutes a codeword in $C_L(D,kG)$ for each $(i,j)$. Suppose that $G_{C_L}$ is the generator matrix of $C_L(D,kG)$. There exist vectors ${\bf v}_{i,j}$ of length $\ell(kG)$ satisfying
\[
	{\bf v}_{i,j}G_{C_L}=(h_{i,j}(P_1),\ldots,h_{i,j}(P_R))
\]
for all $(i,j)$. Since the rows of $G_{C_L}$ form a basis of $C_L$, the matrix possesses a right inverse $G_{C_L}^{-1}$. Therefore,
\[
	{\bf v}_{i,j}=(h_{i,j}(P_1),\ldots,h_{i,j}(P_R))G_{C_L}^{-1}.
\]
The master node can thus decode the collected results to recover $AB$ by extracting the coefficients of $h$ from the coordinates of ${\bf v}_{i,j}$.

The computation of $G_{C_L}^{-1}$ incurs a total cost $O(\ell(kG)^2R)$. For polynomial AG codes, recovering all coefficients of $h$ requires computing every ${\bf v}_{i,j}$, resulting in total cost $O(ab\ell(kG)^2)=O(\frac{rt}{mn}\ell(kG)^2)$. For Matdot AG codes, only the $d$-th coordinate of each ${\bf v}_{i,j}$ needs recovery, thus the total cost in this case is $O(ab\ell(kG))=O(rt\ell(kG))$.

The computational complexity at worker nodes must also be considered. For polynomial AG codes, each worker node performs matrix multiplication in $\mathbb{F}_q^{\frac{r}{m}\times s}$ and $\mathbb{F}_q^{s\times \frac{t}{n}}$ with complexity $O(\frac{rst}{mn})$. For Matdot AG codes, each worker node performs matrix multiplication in $\mathbb{F}_q^{r\times\frac{s}{m}}$ and $\mathbb{F}_q^{\frac{s}{m}\times {t}}$ with complexity $O(\frac{rst}{m})$. 

We present the recovery thresholds and complexity comparisons in the following table.
\begin{table}[htbp]\label{table2}
\renewcommand{\arraystretch}{1.5}
\centering
\caption{Results of polynomial AG codes and Matdot AG codes}
\begin{tabular}{|c|c|c|c|}
	\hline
	 & Recovery threshold & Decoding complexity & Worker node\\
	\hline
	Polynomial & $g(F)+mn$ & $O(\frac{rt}{mn}\ell(kG)^2+\ell(kG)^2)$ & $O(\frac{rst}{mn})$\\
	\hline
	Matdot  & $2g(F)+2m-1$ & $O(rt\ell(kG)+\ell(kG)^2)$ &$O(\frac{rst}{m})$\\
	\hline
\end{tabular}
\end{table}

\section{Implementation Results of Polynomial Codes}
In this section, we present the computational results of the matrix multiplication for $AB$ through polynomial codes, where $A\in\mathbb{F}_{25}^{20000\times10000}$ and $B\in\mathbb{F}_{25}^{10000\times12000}$ are random matrices over $\mathbb{F}_{25}$. This comparison indicates that our function fields are faster than rational function fields; therefore, we do not consider Matdot codes. All experiments were conducted in Magma V2.21 running on Windows Subsystem for Linux (WSL2) with Ubuntu 22.04, and utilized an AMD 8845HS processor with 32GB DDR5-6400 MHz RAM. Recalling Example \ref{exp1}, we partition
\begin{equation*}
	A=\begin{pmatrix}
		A_1\\
		A_2\\
		\vdots\\
		A_8
	\end{pmatrix}=\begin{pmatrix}
			A'_1\\
			A'_2\\
			A'_3\\
			A'_4
		\end{pmatrix},B=\begin{pmatrix}
	B_1,B_2,\ldots,B_5
\end{pmatrix}=\begin{pmatrix}
	B'_1,B'_2,\ldots,B'_6
\end{pmatrix}.
\end{equation*}
Since the recovery thresholds satisfy $8 \cdot 5 + 7 = 47 \leq 52$ and $4 \cdot 6 = 24 \leq 25$, we are able to select enough places of the genus 7 function field $\mathbb{F}_{25}(x,y)$ in Example \ref{exp1} and the rational function field $\mathbb{F}_{25}(x')$. The master node must wait for the last worker node to complete its computation, which implies that the maximum computation cost among all $A_iB_j$ (and $A'_iB'_j$) should be taken into account. Table \ref{table3} reports the time cost of direct computation.
\begin{table}[htbp]\label{table3}
\renewcommand{\arraystretch}{1.5}
\centering
\caption{Maximal Time cost of direct matrix multiplication}
\begin{tabular}{|c|c|c|c|}
	\hline
	 & $AB$ & $A_iB_j$ & $A'_iB'_j$\\
	\hline
	Time cost & 87.97s& 2.92s & 5.48s\\
	\hline
\end{tabular}
\end{table}

Then we choose
\[
	f:=\sum\limits_{i=1}^8A_iy^{i-1}, g:=\sum\limits_{j=1}^5B_jx^{j-1},
\]
and
\[
	f':=\sum\limits_{i=1}^4A'_ix'^{i-1}, g':=\sum\limits_{j=1}^6B'_jx'^{4(j-1)}.
\]
The maximal time costs of evaluation and computation for $h(P_i)$ and $h'(P_i)$ are listed in Table \ref{table4}. 
\begin{table}[htbp]\label{table4}
\renewcommand{\arraystretch}{1.5}
\centering
\caption{Maximal Time cost of operations in function fields}
\begin{tabular}{|c|c|c|c|c|c|c|}
	\hline
	 & $f(P_i)$ & $g(P_i)$& $f'(P_i)$ & $g'(P_i)$ & $h(P_i)$ & $h'(P_i)$ \\
	\hline
	Time cost & 0.20s & 0.15s & 0.18s & 0.12s & 3.04s & 4.98s\\
	\hline
\end{tabular}
\end{table}
The implementation results show that the matrix multiplication through our DMM-friendly function fields is marginally faster than rational function fields. Meanwhile, we demonstrate that in this implementation, our DMM-friendly algebraic function field $\mathbb{F}_{25}(x,y)$ exhibits enhanced straggler tolerance compared to the rational function field $\mathbb{F}_{25}(x')$, as it enables the utilization of a larger set of worker nodes. In decoding, both $G_{C_L}$ and $G'_{C_L}$ are sufficiently small matrices, which means the computation cost for their inverses $G_{C_L}^{-1}$ and $G_{C_L}^{'-1}$ can be considered negligible. Our code can be found on GitHub:\href{https://github.com/ZckFreedom/Magma-DMM}{https://github.com/ZckFreedom/Magma-DMM}.

\section{Conclusion and Discussion}
In this paper, we established the optimal recovery threshold bounds for both polynomial AG codes and Matdot AG codes. We provided some explicit constructions of algebraic function fields, that enable efficient implementation of these codes with optimal recovery thresholds.

However, our construction does not achieve the minimal possible genus for DMM with a certain parameters of matrices $A$ and $B$, leaving room for improvement. More precisely, the minimum genus of polynomial AG codes $g_{\text{\rm min}}$ should satisfy
\[
	g_{\text{\rm min}}+mn= q+1+2g_{\text{\rm min}}\sqrt{q}-\epsilon
\]
and the minimum genus of Matdot AG codes $g_{\text{\rm min}}$ should satisfy
\[
	2g_{\text{\rm min}}+2m-1= q+1+2g_{\text{\rm min}}\sqrt{q}-\epsilon
\]
for a small $\epsilon$. Note that our fields have a larger genus, and using functions over these fields in practical polynomial AG codes and Matdot AG codes may be more time-consuming compared to applying the constructions from \cite{Fidalgo} or \cite{Li-Xing} on function fields with smaller genus. Nevertheless, we have demonstrated that our results provide an effective method. Compared to the construction in \cite{Fidalgo}, our fields achieve the optimal recovery threshold. In contrast to \cite{Li-Xing}, although the existence of non-special divisors of degree $g(F)$ has been established \cite{Ballet}, results regarding their explicit construction remain limited \cite{Moreno}.

Another limitation is that our fields are not maximal in most cases. This further restricts their application. Consequently, improved methods for constructing algebraic function fields that are DMM-friendly require investigation, and we identify this problem as a focus for future research. Building on the results of polydot codes provided in \cite{Li-Xing}, the study of optimal recovery threshold bounds will also be addressed in our future work.

Moreover, AG code-based DMM could potentially exhibit enhanced performance in straggler mitigation when employing techniques of locally repairable AG codes with multiple recovering sets \cite{Jin}. We also designate this as an objective for subsequent research.
 
\section{Acknowledgment}
This work is supported by Natural Science Foundation of Guangdong Province (No. 2025A1515011764).


\begin{thebibliography}{1}
\bibliographystyle{IEEEtran}
\bibitem{Aliasgari} M. Aliasgari, O. Simeone, and J. Kliewer, ``Private and secure distributed matrix multiplication with flexible communication load," {\it IEEE Trans. Inf. Forensics Secur.}, vol. 15, pp. 2722–2734, Feb. 2020.


\bibitem{Ballet} S. Ballet, and D. Le Brigand, ``On the existence of non-special divisors of degree $g$ and $g-1$ in algebraic function fields over $\mathbb{F}_q$," in {\it Journal of Number Theory}, vol. 116, no. 2, pp. 293-310, 2006.

\bibitem{Beelen} P. Beelen, J. Rosenkilde and G. Solomatov, ``Fast Decoding of AG Codes," {\it IEEE Trans. Inf. Theory}, vol. 68, no. 11, pp. 7215-7232, Nov. 2022.

\bibitem{Canteaut} A. Canteaut and F. Chabaud, ``A new algorithm for finding minimum-weight words in a linear code: application to McEliece's cryptosystem and to narrow-sense BCH codes of length 511," {\it IEEE Trans. Inf. Forensics Secur.}, vol. 44, no. 1, pp. 367-378, Jan. 1998.

\bibitem{Chang} W.-T. Chang and R. Tandon, ``On the capacity of secure distributed matrix multiplication," in {\it 2018 IEEE Global Communications Conference
(GLOBECOM)}, pp. 1–6, IEEE, 2018.

\bibitem{Chara} M. Chara1, R. Podest{\' a}, L. Quoos, and R. Toledano, ``Lifting iso-dual algebraic geometry codes," {\it Des., Codes Cryptogr}, vol. 92, pp. 2743-2767, May 2024.

\bibitem{Dutta} S. Dutta, V. Cadambe, and P. Grover, ``Coded convolution for parallel and distributed computing within a deadline," in {\it Proc. IEEE Int. Symp. Inf. Theory (ISIT)}, Jun. 2017, pp. 2403–2407.

\bibitem{Dutta-mat} S. Dutta, M. Fahim, F. Haddadpour, H. Jeong, V. Cadambe, and P. Grover, ``On the optimal recovery threshold of coded matrix multiplication,” {\it IEEE Trans. Inf. Theory}, vol. 66, no. 1, pp. 278–301, Jan. 2020.

\bibitem{Fidalgo} A. Fidalgo-D{\' i}az, and U. Mart{\' i}nez-Pe{\~ n}as, ``Distributed Matrix Multiplication With Straggler Tolerance Using Algebraic Function Fields," {\it IEEE Trans. Inf. Theory}, vol. 71, no. 2, pp. 996-1006, Feb. 2025.



\bibitem{Jin} L. Jin, H. Kan, and Y. Zhang, ``Constructions of Locally Repairable Codes With Multiple Recovering Sets via Rational Function Fields," {\it IEEE Trans. Inf. Theory}, vol. 66, no. 1, pp. 202-209, Jan. 2020.

\bibitem{Karpuk} D. Karpuk, and R. Tajeddine, ``Modular Polynomial Codes for Secure and Robust Distributed Matrix Multiplication," {\it IEEE Trans. Inf. Theory}, vol. 70, no. 6, pp. 4396-4413, Jun. 2024.

\bibitem{Khalfaoui} S. E. Khalfaoui, M. Lhotel and J. Nardi, ``Goppa–Like AG Codes From Ca,b Curves and Their Behavior Under Squaring Their Dual," {\it IEEE Trans. Inf. Theory}, vol. 70, no. 5, pp. 3330-3344, May 2024.

\bibitem{Lee} K. Lee, M. Lam, R. Pedarsani, D. Papailiopoulos, and K. Ramchandran, ``Speeding Up Distributed Machine Learning Using Codes," {\it IEEE Trans. Inf. Theory}, vol. 64, no. 3, pp. 1514-1529, Mar. 2018.

\bibitem{Li-Xing} J. Li, S. Li, and C. Xing, “Algebraic geometry codes for distributed matrix multiplication using local expansions,” 2024, {\it arXiv:2408.01806}.

\bibitem{Machado} R. A. Machado, G. L. Matthews, and W. Santos, “HerA scheme: Secure distributed matrix multiplication via Hermitian codes,” in {\it Proc. IEEE Int. Symp. Inf. Theory (ISIT)}, Jun. 2023, pp. 1729–1734.

\bibitem{Makkonen-secure} O. Makkonen and C. Hollanti, ``General Framework for Linear Secure Distributed Matrix Multiplication With Byzantine Servers," {\it IEEE Trans. Inf. Theory}, vol. 70, no. 6, pp. 3864-3877, Jun 2024.

\bibitem{Makkonen} O. Makkonen, E. Sa{\c c}ıkara, and C. Hollanti, ``Algebraic Geometry Codes for Secure Distributed Matrix Multiplication," {\it IEEE Trans. Inf. Theory}, early access, Jan. 29, 2024, doi: \href{https://doi.org/10.1109/TIT.2025.3535091}{10.1109/TIT.2025.3535091}.

\bibitem{Matthews} G. L. Matthews, and P. Soto, ``Algebraic Geometric Rook Codes for Coded Distributed Computing," in {\it 2024 IEEE Information Theory Workshop (ITW)}, Dec. 2024, pp. 717-722.

\bibitem{Mendoza-kummer} E. A.R. Mendoza, ``On Kummer extensions with one place at infinity," {\it Finite Fields Appl.}, vol. 89, pp. 102209, Apr. 2023.

\bibitem{Moreno} E. C. Moreno, H. H. L{\'o}pez, and G. L. Matthews, ``Explicit Non-special Divisors of Small Degree, Algebraic Geometric Hulls, and LCD Codes from Kummer Extensions," {\it SIAM J. Appl. Algebra Geom.}, vol. 8, no. 2, pp. 394-413, May 2024.

\bibitem{Munuera} C. Munuera, A. Sepúlveda, and F. Torres, ``Generalized Hermitian codes," {\it Des., Codes Cryptogr}, vol. 69, pp. 123-130, Mar. 2013.

\bibitem{Navarro} H. Navarro, ``Bases for Riemann–Roch spaces of linearized function fields with applications to generalized algebraic geometry codes," {\it Des., Codes Cryptogr}, vol. 92, pp. 3033-3048, Jun. 2024.

\bibitem{Oliveira} R. G. L. D’ Oliveira, S. El Rouayheb, and D. Karpuk, ``GASP codes for secure distributed matrix multiplication," {\it IEEE Trans. Inf. Theory}, vol. 66, no. 7, pp. 4038–4050, Jul. 2020.

\bibitem{sage} SageMath, the Sage Mathematics Software System (Version 9.8), Sage Developers, Newcastle Upon Tyne, U.K., 2023. [Online]. Available: https://www.sagemath.org

\bibitem{Stich-hermitian} H. Stichtenoth, ``A note on Hermitian codes over GF($q^2$)," {\it IEEE Trans. Inf. Theory}, vol. 34, no. 5, pp. 1345-1348, Sep. 1988.

\bibitem{Stich} H. Stichtenoth, {\it Algebraic Function Fields and Codes} (Graduate Texts in Mathematics), vol. 254. Berlin, Germany: Springer-Verlag, 2009.

\bibitem{Yu-poly} Q. Yu, M. A. Maddah-Ali, and A. S. Avestimehr, ``Polynomial codes: An optimal design for high-dimensional coded matrix multiplication," in {\it Proc. Adv. Neural Inf. Process. Syst.}, vol. 30, Dec. 2017, pp. 4406–4416.

\bibitem{Yu-straggler} Q. Yu, M. A. Maddah-Ali, and A. S. Avestimehr, ``Straggler Mitigation in Distributed Matrix Multiplication: Fundamental Limits and Optimal Coding," {\it IEEE Trans. Inf. Theory}, vol. 66, no. 3, pp. 1920-1933, Mar. 2020.

\end{thebibliography}
\end{document}